\documentclass[
 reprint,
superscriptaddress,
 amsmath,amssymb,
 aps,
pra,
aip,
rsi,]{revtex4-1}
\usepackage{color}
\usepackage{graphicx}
\usepackage{caption}
\usepackage{subcaption}
\usepackage{dcolumn}
\usepackage{bm}
\usepackage{braket}
\usepackage{float}
\usepackage{multirow}
\usepackage{hhline}
\usepackage{mathtools}
\usepackage{qcircuit}
\usepackage{lipsum}
\usepackage{amsmath,amsfonts,amssymb,amscd}
\usepackage{amsthm}
\usepackage{stmaryrd}
\usepackage{algorithm}
\theoremstyle{plain}
\newtheorem{thm}{Theorem}

\theoremstyle{definition}
\newtheorem{defn}[thm]{Definition}

\newtheorem{algo}[algorithm]{Algorithm}

\newcommand{\diag}[1]{\text{diag}(#1)}
\newcommand{\rank}[1]{\text{rank}(#1)}
\newcommand{\row}[1]{\text{row}(#1)}
\newcommand{\col}[1]{\text{col}(#1)}
\newcommand{\tr}[1]{\text{tr}(#1)}

\begin{document}
\title{A Numerical Study of Bravyi-Bacon-Shor and Subsystem Hypergraph Product Codes}
\author{Muyuan Li}
\email{mli97@gatech.edu}
\affiliation{School of Computational Science and Engineering,
 Georgia Institute of Technology, Atlanta, Georgia 30332, USA}
 
 \author{Theodore J. Yoder}
 \email{ted.yoder@ibm.com}
 \affiliation{IBM T. J. Watson Research Center, Yorktown Heights, NY, 10598, United States}
\date{November 2019}

\begin{abstract}
We provide a numerical investigation of two families of subsystem quantum codes that are related to hypergraph product codes by gauge-fixing. The first family consists of the Bravyi-Bacon-Shor (BBS) codes which have optimal code parameters for subsystem quantum codes local in 2-dimensions. The second family consists of the constant rate ``generalized Shor" codes of Bacon and Cassicino \cite{bacon2006quantum}, which we re-brand as subsystem hypergraph product (SHP) codes. We show that any hypergraph product code can be obtained by entangling the gauge qubits of two SHP codes. To evaluate the performance of these codes, we simulate both small and large examples. For circuit noise, a $\llbracket 21,4,3\rrbracket$ BBS code and a $\llbracket 49,16,3\rrbracket$ SHP code have pseudthresholds of $2\times10^{-3}$ and $8\times10^{-4}$, respectively. Simulations for phenomenological noise show that large BBS and SHP codes start to outperform surface codes with similar encoding rate at physical error rates $1\times 10^{-6}$ and $4\times10^{-4}$, respectively.
\end{abstract}

\maketitle

\section{Introduction}

Two-dimensional topological error-correcting codes are extremely attractive models of quantum error-correction. Structurally, low-weight stabilizers -- just weight four for the surface code and weight six for the most popular color code -- that are also local in the plane make for simple fault-tolerant syndrome measurement circuits. In turn, this simplicity leads to surprisingly high thresholds \cite{wang2011surface} compared to, say, concatenated codes \cite{aliferis2005quantum}.

On the other hand, error-correction in two dimensions is inherently limited by the Bravyi-Poulin-Terhal bound \cite{bravyi2010tradeoffs}, which states that a two-dimensional code using $N$ qubits to encode $K$ qubits with code distance $D$ must satisfy $cKD^2\le N$ for some universal constant $c$. In particular, two-dimensional codes with constant rate $K\propto N$ must have constant distance, which precludes error-correction with constant space overhead \cite{gottesman2014fault} in two dimensions.

These constraints on two-dimensional codes explains the recent surge of interest in quantum hypergraph product codes \cite{tillich2013quantum,leverrier2015quantum}, which break the plane (i.e.~are \emph{not} local in two dimensions) but in doing so achieve $K\propto N$ and $D\propto\sqrt{N}$. Given the small-set flip decoder \cite{leverrier2015quantum}, which is single-shot with an asymptotic threshold, hypergraph product codes promise quantum error-correction with constant overhead \cite{fawzi2018constant}.

However, hypergraph product codes also have a couple of undesirable properties from a practical standpoint. First, the small-set flip decoder, although theoretically satisfactory, is likely not practical due to low thresholds even when measurements are perfect \cite{grospellier2018numerical}. This is somewhat to be expected by analogy with classical expander codes, where the classical flip decoder \cite{sipser1996expander} is greatly outperformed by heuristic decoders, such as belief propagation \cite{richardson2001capacity}. It is also unclear that the small-set flip decoder works well at all on small examples suitable for near-term implementation. Second, the stabilizer weights of hypergraph product codes are relatively large, e.g.~the best performing codes in \cite{grospellier2018numerical} have stabilizers with weight 11, which necessitates a corresponding increase in fault-tolerant circuit complexity and a decrease in thresholds with respect to circuit-level noise.

Here we take an empirical look at two families of subsystem codes that, while related to hypergraph product codes, may have some advantages for near-term implementation. Because these are subsystem codes, the operators measured for error-correction are quite small -- in the cases we explore here they never exceed weight six. We also demonstrate how the powerful technique of belief propagation can be applied to decode these codes.

The first family consists of the Bravyi-Bacon-Shor (BBS) codes \cite{bravyi2012subsystem}. BBS codes achieve $K,D\propto\sqrt{N}$ with just two-body measurements and are easily modified so that these measurements are local in two dimensions. Furthermore, they can be gauge-fixed to hypergraph product codes \cite{yoder2019optimal}. The second family consists of the ``generalized Shor" codes of Bacon and Cassacino \cite{bacon2006quantum}. We rename these codes subsystem hypergraph product (SHP) codes, because we prove that any hypergraph product code is two SHP codes with their gauge qubits entangled. SHP codes can achieve $K\propto N$ and $D\propto\sqrt{N}$ just like hypergraph product codes. Compared to BBS codes, they have higher weight gauge operators, weight six in our instances.

We perform numerical experiments with these code families in two regimes of operation. In the small-code regime, we construct small, distance-3 codes in each class, develop fault-tolerant circuits for measuring their stabilizers, and calculate pseudothresholds for circuit noise. We find pseudotresholds of $2\times10^{-3}$ for a $\llbracket21,4,3\rrbracket$ BBS code and $8\times10^{-4}$ for a $\llbracket49,16,3\rrbracket$ SHP code. These results suggest that the BBS code in particular is quite a good candidate for protecting four logical qubits with a small quantum computer.

In the large-code regime, we create BBS and SHP codes from regular classical expander codes. We modify belief propagation to include measurement errors and apply it to decode these codes under an error model including data and measurement noise (but without circuit-level noise). Despite no asymptotic thresholds, compared to a single logical qubit of surface code with similar encoding rate, BBS and SHP codes do achieve better logical error rates per logical qubit provided sufficiently low physical error rates: $p<10^{-6}$ for BBS codes and $p<4\times10^{-4}$ for SHP codes.

The paper is organized as follows. In Section \ref{sec:BBS} we review the Bravyi-Bacon-Shor codes and present a circuit-level simulation of the $\llbracket21,4,3\rrbracket$ code. In Section \ref{sec:bigSHP} we look at the construction of the subsystem hypergraph product codes, find their code parameters, and present a circuit-level simulation of the $\llbracket49,16,3\rrbracket$ code. In Section \ref{sec:decode} we show how to add measurement noise to a classical belief propagation decoder so that it can then be used to decode the BBS and SHP codes. In Section \ref{sec:results} we present numerical results on large BBS and SHP codes and compare them to surface codes.





\section{Review of Bravyi-Bacon-Shor Codes}
\label{sec:BBS}
In this section, we review the Bravyi-Bacon-Shor (BBS) codes that were introduced by Bravyi \cite{bravyi2010tradeoffs} and explicitly constructed in \cite{yoder2019optimal}.

Let $\mathcal{F}_2$ denote the finite field with two elements 0,1. A Bravyi-Bacon-Shor code is defined by a binary matrix $A \in \mathcal{F}_2^{n_1 \times n_2}$, where qubits live on sites $(i,j)$ of the matrix $A$ for which $A_{i,j} = 1$. As shown in \cite{bravyi2010tradeoffs,yoder2019optimal}, given $A$ we can define two classical codes corresponding to its column-space and row-space:
\begin{align}
    \mathcal{C}_1 &= \col{A},\\
    \mathcal{C}_2 &= \row{A},
\end{align}
where $\mathcal{C}_1$ and $\mathcal{C}_2$ has code parameters $[n_1, k, d_1]$, $[n_2, k, d_2]$, generating matrices $G_1$ and $G_2$, and parity check matrices $H_1$ and $H_2$.

The notation for Pauli operators on the qubit lattice is defined as follows. A Pauli $X$- or $Z$-type operator acting on the qubit at site $(i,j)$ in the lattice is written as $X_{i,j}$ or $Z_{i,j}$. A Pauli operator acting on multiple qubits is specified by its support $S$:
\begin{equation}
    X(S) = \prod_{ij}(X_{i,j})^{S_{ij}}, \,\, S \in \mathcal{F}_2^{n_1 \times n_2},
\end{equation}
where $S_{ij}=1$ implies that $A_{ij}=1$, since qubits only exist where $A_{ij}=1$. Similar notations will be used throughout the rest of this paper. We let $|A|=\sum_{ij}A_{ij}$ and $|v|=\sum_iv_i$ denote the Hamming weights of matrices and vectors.

\begin{defn}\cite{bravyi2010tradeoffs}
The Bravyi-Bacon-Shor code constructed from $A \in \mathcal{F}_2^{n_1 \times n_2}$, denoted $\text{BBS}(A)$, is an $\llbracket N,K,D\rrbracket$ quantum subsystem code with gauge group generated by 2-qubit operators and 
\begin{align*}
    N&=|A|,\\
    K&=\rank{A},\\
    D&=\min\{ |\vec{y}|>0:\vec{y}\in \row{A} \cup \col{A}\},\\
\end{align*}
\end{defn}
As CSS quantum subsystem codes, the gauge group of BBS codes is generated by $XX$ interactions between any two qubits sharing a column in $A$ and $ZZ$ between any two qubits sharing a row in $A$. The gauge group can be more formally written as 
\begin{align}
\label{eq:BBS_gx}
\mathcal{G}^{(\text{bbs})}_X&=\{X(S):G_R S=0,S \subseteq A\},\\
\label{eq:BBS_gz}
\mathcal{G}^{(\text{bbs})}_Z&=\{Z(S):S G_R^T =0,S \subseteq A\},
\end{align}
where $G_R = (1,1, \ldots, 1)$ is the generating matrix of the classical repetition code, and the subset notation $S\subseteq A$ means that $S$ is a matrix such that, for all $i,j$, $S_{ij}=1$ implies $A_{ij}=1$.

For bare logical operators of the BBS code to commute with all of its gauge operators, each bare logical X-type operator must be supported on entire rows of the matrix and each bare logical Z-type operator must be supported on entire columns of the matrix. To express this similarly to the gauge operators above, define the parity check matrix of the classical repetition code $H_R$. Then we have the sets of $X$- and $Z$-type logical operators:
\begin{align}
\label{eq:BBS_Lx}
\mathcal{L}^{(\text{bbs})}_X&=\{X(S \cap A):S H^T_R=0\},\\
\label{eq:BBS_Lz}
\mathcal{L}^{(\text{bbs})}_Z&=\{Z(S \cap A):H_R S =0\}.
\end{align}

Consequently, the group of stabilizers for the BBS code is the intersection of the group of bare logical operators with the gauge group:
\begin{align}
\label{eq:BBS_Sx}
\mathcal{S}_X^{(\text{bbs})} &= \mathcal{L}_X^{(\text{bbs})} \cap \mathcal{G}_X^{(\text{bbs})}\\
    &= \{ X(S \cap A):S H^T_R=0, G_1 S=0 \},\\
\label{eq:BBS_Sz}
\mathcal{S}_Z^{(\text{bbs})} &= \mathcal{L}_Z^{(\text{bbs})} \cap \mathcal{G}_Z^{(\text{bbs})}\\
    &= \{ Z(S \cap A):H_R S=0, S G_2^T=0 \}.
\end{align}

\subsection{Constructing BBS codes with classical linear codes}
In \cite{yoder2019optimal}, the following method of constructing a BBS code from classical codes was given.
\begin{thm}
 Given two classical linear codes $\mathcal{C}_1$ and $\mathcal{C}_2$ with parameters $[ n_1,k,d_1 ]$ and $[ n_2,k,d_2]$, and generating matrices $G_1 \in \mathcal{F}_2^{k \times n_1}$ and $G_2 \in \mathcal{F}_2^{k \times n_2}$, we can construct the code $BBS(A)$ by
\begin{equation}
    A = G^T_1 Q G_2 \in \mathcal{F}_2^{n_1 \times n_2},
\end{equation}
where $Q \in \mathcal{F}_2^{k \times k}$ can be any full rank $k \times k$ matrix. Then $BBS(A)$ is an $\llbracket N,K,D \rrbracket$ quantum subsystem code with 
\begin{align}
    min(n_1 d_2, d_1 n_2) &\leq N \leq n_1 n_2,\\
    K &= k,\\
    D &= min(d_1, d_2).
\end{align}
\end{thm}
The matrix $Q \in \mathcal{F}_2^{k \times k}$ represents the non-uniqueness of the generating matrices, and adjusting $Q$ would only affect the number of physical qubits in $\text{BBS}(A)$. It is easy to see that $\col{A} = \row{G_1} = \mathcal{C}_1$, and $\row{A} = \row{G_2} = \mathcal{C}_2$, and the conclusions in the theorem about the code parameters follow.

\subsection{Example: A $[[21,4,3]]$ Bravyi-Bacon-Shor Code}
\label{sec:BBS_Hamming}
\begin{table*}
\begin{center}
\begin{tabular}{ c|c|c||c}
 \hline
 Qubits & $X_L$ & $Z_L$ & Stabilizers\\ \hline \hline
1 & $X_0 X_1 X_2$ & $Z_0 Z_12 Z_{17}$ & $X_0 X_1 X_2 X_3 X_4 X_5 X_9 X_{10} X_{11} X_{12} X_{13} X_{14}$\\
2 & $X_3 X_4 X_5$ & $Z_3 Z_9 Z_{16}$ & $X_0 X_1 X_2 X_6 X_7 X_8 X_9 X_{10} X_{11} X_{15} X_{16} X_{17}$\\ 
3 & $X_6 X_7 X_8$ & $Z_6 Z_{15} Z_{18}$ & $X_3 X_4 X_5 X_6 X_7 X_8 X_9 X_{10} X_{11} X_{18} X_{19} X_{20}$\\ 
4 & $X_3 X_4 X_5 X_6 X_7 X_8 X_9 X_{10} X_{11}$ & $Z_3 Z_4 Z_9 Z_{13} Z_{16} Z_{19}$ & $Z_6 Z_{15} Z_{18} Z_3 Z_9 Z_{16} Z_4 Z_{13} Z_{19} Z_7 Z_{10} Z_{14}$\\  
& &  & $Z_6 Z_{15} Z_{18} Z_0 Z_{12} Z_{17} Z_{4} Z_{13} Z_{19} Z_{1} Z_5 Z_8$\\
& &  & $Z_3 Z_{9} Z_{16} Z_0 Z_{12} Z_{17} Z_{4} Z_{13} Z_{19} Z_{2} Z_{11} Z_{20}$\\  
\hline \hline   
\end{tabular}
\caption{Stabilizers and a set of canonical logical operators for the $[[21, 4, 3]]$ Bravyi-Bacon-Shor code constructed using the $[7,4,3]$ Hamming code.}
\label{table:hamming}
\end{center} 
\end{table*}

The $[7,4,3]$ Hamming code is generated by \[G = 
    \begin{bmatrix}
    1 & 0 & 0 & 0 & 1 & 1 & 0 \\
    0 & 1 & 0 & 0 & 1 & 0 & 1 \\
    0 & 0 & 1 & 0 & 0 & 1 & 1 \\
    0 & 0 & 0 & 1 & 1 & 1 & 1
    \end{bmatrix}
    ,
    \text{\space}H = 
    \begin{bmatrix}
    1 & 1 & 0 & 1 & 1 & 0 & 0 \\
    1 & 0 & 1 & 1 & 0 & 1 & 0 \\
    0 & 1 & 1 & 1 & 0 & 0 & 1
    \end{bmatrix}.
    \]
    Using 
    $Q = \begin{psmallmatrix} 0 & 0 & 1 & 0\\ 0 & 1 & 0 & 1 \\ 1 & 0 & 0 & 0 \\ 0 & 1 & 0 & 0\end{psmallmatrix}$ we can construct a $[[21, 4, 3]]$ Bravyi-Bacon-Shor code $A = G^T Q G$:
    \[A = 
    \begin{bmatrix}
    0 & 0 & 1 & 0 & 0 & 1 & 1 \\
    0 & 1 & 0 & 1 & 0 & 1 & 0 \\
    1 & 0 & 0 & 0 & 1 & 1 & 0 \\
    0 & 1 & 0 & 0 & 1 & 0 & 1 \\
    0 & 0 & 1 & 1 & 1 & 0 & 0 \\
    1 & 1 & 1 & 0 & 0 & 0 & 0 \\
    1 & 0 & 0 & 1 & 0 & 0 & 1 \\
    \end{bmatrix},
    \]
which minimizes the number of qubits used. We construct a canonical set of  bare logical operators for the four logical qubits encoded in the $\llbracket 21,4,3 \rrbracket$ code along with a set of stabilizers generators, as shown in TABLE \ref{table:hamming}. Note that while qubit $4$ has high weight bare logical operators due to the construction that we have chosen, it can still suffer from weight three logical operators, such as $Z_4Z_{13}Z_{19}$, and so its error rate has the same slope as the others.

We estimated the performance of this code by simulating it under circuit level standard depolarizing noise, where Pauli channels with Kraus operators
\begin{align}
\begin{split}
E_{1q} &= \{\sqrt{1-p}I, \sqrt{\frac{p}{3}}X, \sqrt{\frac{p}{3}}Y, \sqrt{\frac{p}{3}}Z\},\\
E_{2q} &= \{\sqrt{1-p}I, \sqrt{\frac{p}{15}}IX, \ldots \sqrt{\frac{p}{15}}ZZ\},
\end{split}
\label{eq:pauli}
\end{align}
are applied after each 1- and 2-qubit gate in the circuit, respectively. We call $p\in[0,1]$ the physical error rate. Assuming the code is fault-tolerantly prepared into its logical $\ket{0000}$ state, we simulated the circuit of error correction and destructive measurement of data qubits with single qubit memory errors added before error correction. The same error rate is used across the circuit for memory errors, gate errors, and measurement errors. Note that idle errors are not considered in the circuit-level simulations presented in this paper. When we consider a trapped ion architecture where long-range interactions required by these subsystem codes of interest can be easily implemented, idle errors have minimal effect to the logical system when compared to gate errors \cite{debroy2019logical}.

The results are shown in FIG.~\ref{fig:BBS_hamming}. Note that since the BBS codes can be considered a compass code in 2-dimensions \cite{li20192d}, to create a fault-tolerant circuit for syndrome extraction it suffices to use a single ancillary qubit for each of the weight-12 stabilizers as listed in TABLE \ref{table:hamming}. Hence the total number of qubits required to perform fault-tolerant syndrome extraction for this code is $21+6=27$. We perform the syndrome extraction once and if the syndrome is trivial, we stop and no correction is needed. If the syndrome is not trivial, we measure the syndrome again and decode with the outcome.

From FIG.~\ref{fig:BBS_hamming} we can see that qubit $4$ performs slightly worse than qubits 1-3, due to the fact that its higher weight logical operators have more chance of anti-commuting with dressed logical operators than qubits 1-3.

\begin{figure}[ht]
    \centering
    \includegraphics[width = 0.9\linewidth]{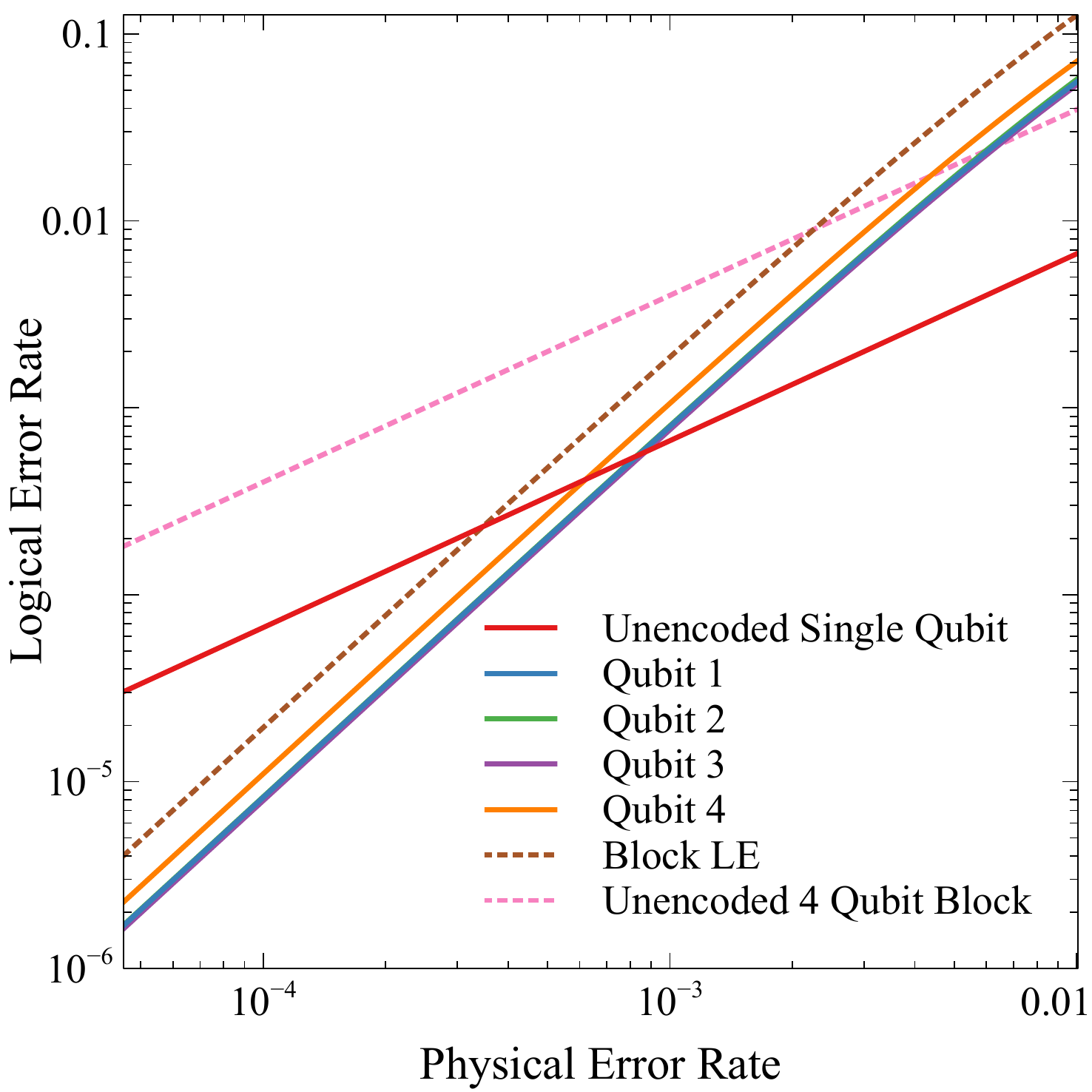}
    \caption{Simulated performance of the $[[21, 4, 3]]$ BBS code under circuit level depolarizing error, with one ancillary qubit per stabilizer for fault-tolerant syndrome extraction. The block pseudothreshold for the code block with 4 encoded logical qubits is $2.3 \times 10^{-3}$, while the per logical qubit pseudothreshold for qubits 1-3 is $8.7 \times 10^{-4}$.}
    \label{fig:BBS_hamming}
\end{figure}
    
\section{Another family of subsystem hypergraph product codes}
\label{sec:bigSHP}
In this section we take a look at the ``generalized Shor" codes in Bacon and Casaccino \cite{bacon2006quantum} from a new perspective. In particular, we find that these codes are in a sense the most natural subsystem hypergraph product codes because two of them, without ancillas, can be gauge-fixed to a hypergraph product code and, conversely, any hypergraph product code can be gauge-fixed into two generalized Shor codes. We therefore refer to generalized Shor codes as subsystem hypergraph product (SHP) codes.

Contrast SHP codes with BBS codes, which can also be gauge-fixed to hypergraph product codes \cite{yoder2019optimal}. Gauge-fixing BBS codes requires ancillas and the result is only a certain subset of all hypergraph product codes with less than constant rate.

\subsection{Hypergraph product codes}
\label{sec:HGP}
To facilitate our proofs, we review the hypergraph product code construction briefly in this section.

\begin{defn}\cite{tillich2013quantum}
Let $H_1\in\{0,1\}^{n_1^T\times n_1}$ and $H_2\in\{0,1\}^{n_2^T\times n_2}$. The hypergraph product (HGP) of $H_1$ and $H_2$ is a quantum code $\text{HGP}(H_1,H_2)$ with stabilizers
\begin{align}\label{eq:hgp_sx}
S^{(\text{hgp})}_X&=\left(H_1\otimes I_{n_2},I_{n_1^T}\otimes H_2^T\right),\\\label{eq:hgp_sz}
S^{(\text{hgp})}_Z&=\left(I_{n_1}\otimes H_2,H_1^T\otimes I_{n_2^T}\right).
\end{align}
\end{defn}
By Eq.~\eqref{eq:hgp_sx} we mean that each vector $v\in\mathcal{F}_2^N$ in the rowspace of the matrix on the righthand side indicates an $X$-type Pauli operator $X^{v}:=\prod_{i=1}^NX_i^{v_i}$ in the stabilizer group. Likewise with $Z$-type operators in Eq.~\eqref{eq:hgp_sz}. Similar notation will be used throughout this section.

Treating $H_1$ and $H_2$ as parity check matrices, we have two classical codes $\mathcal{C}_1=\ker(H_1)$ and $\mathcal{C}_2=\ker(H_2)$ with some parameters $[n_1,k_1,d_1]$ and $[n_2,k_2,d_2]$, respectively. Likewise, treat $H_1^T$ and $H_2^T$ as parity check matrices of the ``transpose" codes $\mathcal{C}_1^T=\ker(H_1^T)$ and $\mathcal{C}_2^T=\ker(H_2^T)$ with respective parameters $[n_1^T,k_1^T,d_1^T]$ and $[n_2^T,k_2^T,d_2^T]$. Because of the rank-nullity theorem
\begin{equation}\label{eq:parameter_relation}
n_i-k_i=n_i^T-k_i^T
\end{equation}
for $i=1,2$. The hypergraph product code $\text{HGP}(H_1,H_2)$ then has parameters \cite{tillich2013quantum}
\begin{equation}
\llbracket n_1n_2+n_1^Tn_2^T,k_1k_2+k_1^Tk_2^T,D\rrbracket,
\end{equation}
where
\begin{equation}
D=\bigg\{\begin{array}{lr}\min(d_1,d_2),&k_1^T=0\text{ or }k_2^T=0\\\min(d_1,d_2,d_1^T,d_2^T),&\text{otherwise}\end{array}.
\end{equation}

Moving on, we notice that there are $n_1n_2+n_1^Tn_2^T$ qubits in $\text{HGP}(H_1,H_2)$ that we lay out on two square lattices, an $n_1\times n_2$ lattice referred to as the ``large'' lattice, denoted $L$, and an $n_1^T\times n_2^T$ lattice referred to as the ``small'' lattice, denoted $l$. See Fig.~\ref{fig:hypergraph_prod.png}. Despite the names, the small lattice need not contain fewer qubits than the large lattice, although typically (e.g.~in random constructions of classical LDPC codes \cite{Gallager1962}) $n_i^T\approx n_i-k_i<n_i$ and this is the case.
\begin{figure}
\centering
\includegraphics[width=0.5\columnwidth]{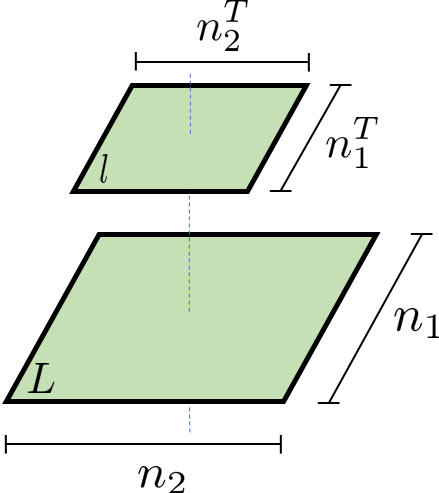}
\caption{The large and small lattices, $L$ and $l$.}
\label{fig:hypergraph_prod.png}
\end{figure}

We label qubits in these lattices in row major fashion. Thus, a (row) vector $r^T\otimes c^T$ for $r\in\{0,1\}^{n_1}$ and $c\in\{0,1\}^{n_2}$ indicates exactly the qubits that are both in the rows indicated by $r$ and in the columns indicated by $c$ of the large lattice. Qubits in the large lattice are labeled first, i.e.~$1,2,\dots,n_1n_2$, followed by qubits in the small lattice, i.e.~$n_1n_2+1,\dots,n_1n_2+n_1^Tn_2^T$.

For later purposes, we point out some subgroups of the stabilizer group. For instance, certain stabilizers of $\text{HGP}(H_1,H_2)$ are supported entirely on the large lattice. Because the rowspace of $S^{(\text{hgp})}_X$ represents all $X$-type stabilizers, if $x\in\{0,1\}^{n_1^T}$, $c\in\ker(H_2)=\mathcal{C}_2\subseteq\{0,1\}^{n_2}$, then
\begin{equation}
(x^T\otimes c^T)S^{(\text{hgp})}_X=\left(x^TH_1\otimes c^T,0\right)
\end{equation}
is a stabilizer supported entirely on the first $n_1n_2$ qubits, i.e.~entirely on the large lattice. Let $G_1\in\{0,1\}^{k_1\times n_1}$ and $G_2\in\{0,1\}^{k_2\times n_2}$ be generator matrices for codes $\mathcal{C}_1$ and $\mathcal{C}_2$. Then, we can provide a generating set of stabilizers on the large lattice like
\begin{align}\label{eq:large_lattice_Sx}
S^{(\text{hgp},L)}_X&=\left(H_1\otimes G_2\right),\\\label{eq:large_lattice_Sz}
S^{(\text{hgp},L)}_Z&=\left(G_1\otimes H_2\right).
\end{align}

Similarly, some stabilizers of $\text{HGP}(H_1,H_2)$ are supported entirely on the small lattice. Let $F_1\in\{0,1\}^{k_1^T\times n_1^T}$ and $F_2\in\{0,1\}^{k_2^T\times n_2^T}$ be generating matrices for codes $\mathcal{C}_1^T$ and $\mathcal{C}_2^T$. The stabilizers on the small lattice have generating sets
\begin{align}\label{eq:small_lattice_Sx}
S^{(\text{hgp},l)}_x&=\left(F_1\otimes H_2^T\right),\\\label{eq:small_lattice_Sz}
S^{(\text{hgp},l)}_z&=\left(H_1^T\otimes F_2\right).
\end{align}

Logical operators of $\text{HGP}(H_1,H_2)$ are those that commute with all stabilizers (we include the stabilizers themselves in this count). For instance, rows of the matrix $\left(I_{n_1}\otimes G_2,0\right)$ indicate $X$-type logical operators, since
\begin{equation}
S^{(\text{hgp})}_Z\left(I_{n_1}\otimes G_2,0\right)^T=0.
\end{equation}
The complete generating sets of $X$-type and $Z$-type logical operators are
\begin{align}
L^{(\text{hgp})}_X&=\left(\begin{array}{cc}
H_1\otimes I_{n_2}&I_{n_1^T}\otimes H_2^T\\
I_{n_1}\otimes G_2&0\\
0&F_1\otimes I_{n_2^T}
\end{array}\right),\\
L^{(\text{hgp})}_Z&=\left(\begin{array}{cc}
I_{n_1}\otimes H_2&H_1^T\otimes I_{n_2^T}\\
G_1\otimes I_{n_2}&0\\
0&I_{n_1^T}\otimes F_2
\end{array}\right).
\end{align}
Nontrivial logical operators are logical operators that are not stabilizers.

An alternative representation of stabilizers and logical operators is to specify them by their supports. For instance $X^{(L)}(S)$ is an $X$-type Pauli supported on the qubits specified by $S\in\{0,1\}^{n_1\times n_2}$ in the large lattice. Likewise for $X^{(l)}(T)$ with $T\in\{0,1\}^{n_1^T\times n_2^T}$ on the small lattice. Of course, $Z$-type Paulis $Z^{(L)}(S)$, $Z^{(l)}(T)$ are denoted analogously.

Using this support-matrix notation, we get alternative descriptions of the stabilizer groups
\begin{alignat}{3}\label{eq:ShgpX}
\mathcal{S}^{(\text{hgp})}_X&=\big\{X^{(L)}(S)X^{(l)}(T):&SH_2^T&=H_1^TT,\\\nonumber
&&G_1S&=0,TF_2^T=0\big\},\\\label{eq:ShgpZ}
\mathcal{S}^{(\text{hgp})}_Z&=\big\{Z^{(L)}(S)Z^{(l)}(T):&H_1S&=TH_2,\\\nonumber
&&SG_2^T&=0,F_1T=0\big\}.
\end{alignat}
and the logical operators
\begin{align}\label{eq:LhgpX}
\mathcal{L}_X^{(\text{hgp})}&=\{X^{(L)}(S)X^{(l)}(T):SH_2^T=H_1^TT\},\\\label{eq:LhgpZ}
\mathcal{L}_Z^{(\text{hgp})}&=\{Z^{(L)}(S)Z^{(l)}(T):H_1S=TH_2\},
\end{align}
which are useful for discussing gauge-fixing later.

\subsection{Subsystem hypergraph product codes}
\label{sec:SHP}

In this section, we define the generalized Shor codes from \cite{bacon2006quantum} with notation similar to our description of HGP codes. This makes the two code families easier to relate later.

\begin{defn}\label{def:SHP_codes}
Let $H_1\in\{0,1\}^{n_1^T\times n_1}$ and $H_2\in\{0,1\}^{n_2^T\times n_2}$. The subsystem hypergraph product (SHP) code of $H_1$ and $H_2$ is the quantum subsystem code $\text{SHP}(H_1,H_2)$ with gauge operators
\begin{align}
G^{(\text{shp})}_X&=\left(H_1\otimes I_{n_2}\right),\\
G^{(\text{shp})}_Z&=\left(I_{n_1}\otimes H_2\right).
\end{align}
\end{defn}
It is worth noting that while the definition of $\text{HGP}(H_1,H_2)$ depends on the parity check matrices $H_1$ and $H_2$, the definition of $\text{SHP}(H_1,H_2)$ depends only on the codes $\mathcal{C}_1=\ker{H_1}$ and $\mathcal{C}_2=\ker{H_2}$. This is because the gauge groups $G^{(\text{shp})}_X$ and $G^{(\text{shp})}_Z$ are the same for $\text{SHP}(H_1,H_2)$ and $\text{SHP}(H_1',H_2')$ whenever $\row{H_1}=\row{H_1'}$ and $\row{H_2}=\row{H_2'}$.

Let us calculate the parameters $\llbracket N,K,D\rrbracket$ of the SHP code. There are clearly $N=n_1n_2$ qubits in the code, which we place on a lattice like in Fig.~\ref{fig:SHP_code}.

\begin{figure}
\centering
\includegraphics[width=0.5\columnwidth]{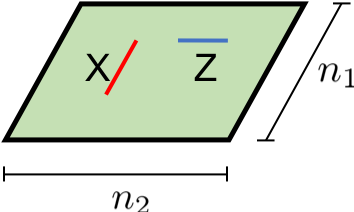}
\caption{A subsystem hypergraph product code. For each column, $X$-type gauge operators are supported on qubits indicated by the parity checks $H_1$. For each row, $Z$-type gauge operators are supported on qubits indicated by the parity checks $H_2$.}
\label{fig:SHP_code}
\end{figure}

To calculate $K$, begin by noticing that certain $X$-type operators, the bare $X$-type logical operators, commute with the entire group of gauge operators. These are generated by
\begin{equation}\label{eq:LX}
L^{(\text{shp})}_X=\left(I_{n_1}\otimes G_2\right),
\end{equation}
because $L^{(\text{shp})}_X\left(G^{\text{(shp)}}_Z\right)^T=0$. Likewise, the bare $Z$-type logical operators are
\begin{equation}\label{eq:LZ}
L^{(\text{shp})}_Z=\left(G_1\otimes I_{n_2}\right).
\end{equation}
The stabilizers of a subsystem code are those gauge operators that also commute with all elements of the gauge group, i.e.~the center of the gauge group. These are generated by
\begin{align}\label{eq:shp_SX}
S^{(\text{shp})}_X&=\left(H_1\otimes G_2\right),\\\label{eq:shp_SZ}
S^{(\text{shp})}_Z&=\left(G_1\otimes H_2\right),
\end{align}
matching those stabilizers of $\text{HGP}(H_1,H_2)$ that are supported entirely on the large lattice (see Eqs.~\eqref{eq:large_lattice_Sx}, \eqref{eq:large_lattice_Sz}). 

Next, the number of encoded qubits can be calculated by comparing the ranks of $L^{(\text{shp})}_X$ and $S^{(\text{shp})}_X$ (or, equivalently of $L^{(\text{shp})}_Z$ and $S^{(\text{shp})}_Z$).
\begin{align}
K&=\rank{L^{(\text{shp})}_X}-\rank{S^{(\text{shp})}_X}\\
&=n_1k_2-(n_1-k_1)k_2\\
&=k_1k_2.
\end{align}

What does the description of $\text{SHP}(H_1,H_2)$ look like in support-matrix notation? Writing down the relevant groups, we have
\begin{align}
\mathcal{G}^{(\text{shp})}_X&=\{X(S):G_1S=0\},\\
\mathcal{G}^{(\text{shp})}_Z&=\{Z(S):SG_2^T=0\},\\
\mathcal{L}^{(\text{shp})}_X&=\{X(S):SH_2^T=0\},\\
\mathcal{L}^{(\text{shp})}_Z&=\{Z(S):H_1S=0\},\\
\mathcal{S}^{(\text{shp})}_X&=\{X(S):G_1S=0,SH_2^T=0\},\\
\mathcal{S}^{(\text{shp})}_Z&=\{Z(S):SG_2^T=0,H_1S=0\}.
\end{align}
Dressed logical operators are denoted $\hat{\mathcal{L}}^{(\text{shp})}_X=\mathcal{L}^{(\text{shp})}_X\mathcal{G}^{(\text{shp})}_X$ and $\hat{\mathcal{L}}^{(\text{shp})}_Z=\mathcal{L}^{(\text{shp})}_Z\mathcal{G}^{(\text{shp})}_Z$.

To compute the distance $D$ of the subsystem hypergraph product code, we need to find the minimum weight of an element of $\hat{\mathcal{L}}^{(\text{shp})}_X-\mathcal{G}^{(\text{shp})}_X$ or of $\hat{\mathcal{L}}^{(\text{shp})}_Z-\mathcal{G}^{(\text{shp})}_Z$. Let us suppose $M\in\hat{\mathcal{L}}^{(\text{shp})}_X-\mathcal{G}^{(\text{shp})}_X$. Then, $M$ can be written as $M=X(S)X(T)$ where $X(S)\in\mathcal{L}^{(\text{shp})}_X$ and $X(T)\in\mathcal{G}^{(\text{shp})}_X$, so $SH_2^T=0$ and $G_1T=0$. Also, since $M$ is not in $\mathcal{G}^{(\text{shp})}_X$, there is some $M'$ corresponding to a row of $L_Z^{(\text{shp})}$ that anticommutes with $M$. Glancing at Eq.~\eqref{eq:LZ}, this means $M'=X(S')$ where $S'$ is the outer product $S'=\vec c\hspace{2pt}\hat e_j^T$ for some $\vec c\in\mathcal{C}_1$ and some $j$ such that
\begin{equation}
\tr{((S+T)^TS')}=\hat e_j^T(S+T)^T\vec c=1.
\end{equation}
This trace being 1 (modulo two) expresses the anticommutation of $M$ and $M'$. Clearly, it implies $(S+T)^T\vec c\neq\vec 0$. Because $\vec c\in\mathcal{C}_1$, there is a vector $\vec x$ such that $\vec c=G_1^T\vec x$ and accordingly,
\begin{equation}
(S+T)^T\vec c=S^T\vec c+T^TG_1^T\vec x=S^T\vec c
\end{equation}
using $G_1T=0$. Moreover, $H_2S^T\vec c=0$ using $SH_2^T=0$ and so $(S+T)^T\vec c$ is a nonzero vector in $\ker(H_2)=\mathcal{C}_2$. Thus, by definition of the classical code distance $|M|=|S+T|\ge|(S+T)^T\vec c\hspace{1pt}|\ge d_2$.

Likewise, if we suppose $M\in\hat{\mathcal{L}}^{(\text{shp})}_Z-\mathcal{G}^{(\text{shp})}_Z$ we find $|M|\ge d_1$. Thus, we have shown $D\ge\min(d_1,d_2)$ and it is not hard given the form of $L^{(\text{shp})}_X$ and $L^{(\text{shp})}_Z$ to see that this in fact holds with equality $D=\min(d_1,d_2)$. Therefore, the subsystem hypergraph product code is a $\llbracket n_1n_2,k_1k_2,\min(d_1,d_2)\rrbracket$ code.

Quantum subsystem codes generalize quantum subspace codes because their stabilizers and logical qubits do not fix all the available degrees of freedom. The remaining degrees of freedom are counted as gauge qubits. These can be thought of as extra logical qubits that are not protected and thus not used to hold any meaningful information. If we calculate the number of gauge qubits in a subsystem hypergraph product code, we find it is
\begin{equation}\label{eq:count_gauge_qubits}
N-\rank{S^{(\text{shp})}_X}-\rank{S^{(\text{shp})}_Z}-K=(n_1-k_1)(n_2-k_2).
\end{equation}

\subsection{Example: A $\llbracket 49, 16, 3 \rrbracket$ subsystem hypergraph product code}
\begin{figure}
    \centering
    \includegraphics[width = 0.9\linewidth]{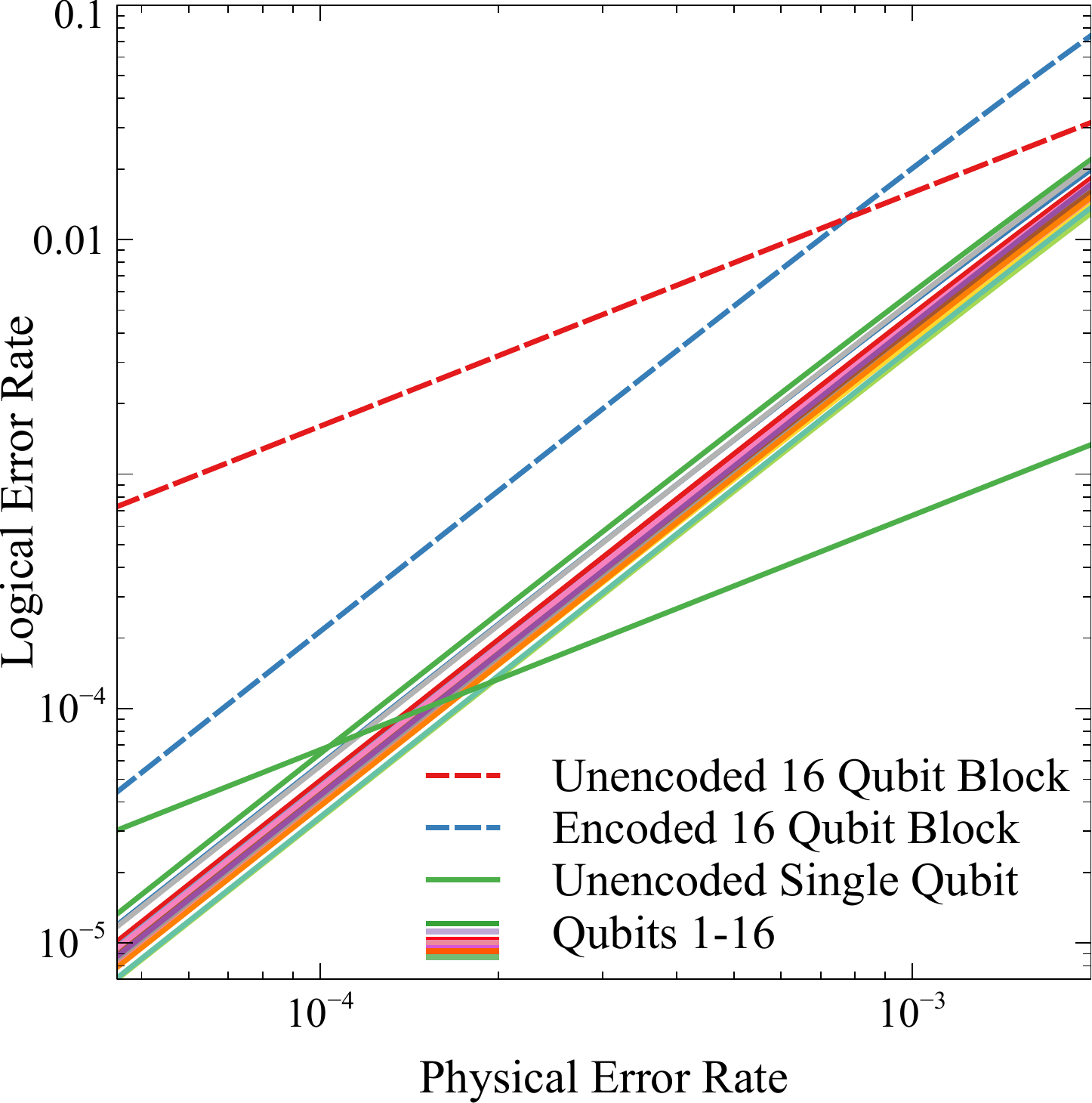}
    \caption{Simulated performance of the $\llbracket 49, 16, 3 \rrbracket$ SHP code under circuit level depolarizing noise, with one ancillary qubit per stabilizer for fault-tolerant syndrome extraction. The block pseudothreshold for the single code block with 16 encoded logical qubits is $8 \times 10^{-4}$, while the per logical qubit pseudothreshold ranges between $10^{-4}$ to $2 \times 10^{-4}$.}
    \label{fig:SHP_hamming}
\end{figure}
Using the classical $[7,4,3]$ Hamming code for both the $X$ and $Z$ part, we can construct a $\llbracket 49, 16, 3 \rrbracket$ subsystem hypergraph product code by following Definition~\ref{def:SHP_codes}. We can construct a canonical set of logical operators of weight $3$ and $4$ for the $16$ logical qubits encoded in the same code block, along with a set of 24 stabilizer generators. Note that similar to the BBS codes, for each of the stabilizers it suffices to use a single ancillary qubit to fault-tolerantly extract its syndrome, by performing CNOT gates in the order of gauge operators and hence directing propagated errors away from the direction of logical errors.

Similar to the $\llbracket 21, 4, 3 \rrbracket$ BBS code, we study the $\llbracket 49, 16, 3 \rrbracket$ SHP code under circuit level depolarizing noise as shown in Eq.~\ref{eq:pauli}. The results are shown in FIG. \ref{fig:SHP_hamming}. Since the $4$ encoded logical bits in the $[7,4,3]$ Hamming code have different performances and the constructed logical operators for the SHP code have different weights, the performance of the $16$ encoded logical qubits varies and their pseudothreshold ranges between $10^{-4}$ and $2 \times 10^{-4}$.

\subsection{SHP codes gauge-fix to HGP codes}
To begin, we define gauge-fixing in general. See also \cite{yoder2019optimal}. We use the notation that for gauge group $\mathcal{G}$, its stabilizer group (centralizer) is $\mathcal{S}(\mathcal{G})$ and it encodes $K(\mathcal{G})$ qubits.
\begin{defn}\label{defn:gauge_fixing}
We say that the gauge group $\mathcal{G}'$ is a gauge-fixing of the gauge group $\mathcal{G}$ if
\begin{enumerate}
\item $\mathcal{S}(\mathcal{G})\le\mathcal{S}(\mathcal{G}')\le\mathcal{G}'\le\mathcal{G}$, and
\item $K(\mathcal{G})=K(\mathcal{G}')$.
\end{enumerate}
We also say that a code is a gauge-fixing of another code if their gauge groups are related in this way.
\end{defn}

We noted below Eq.~\eqref{eq:shp_SZ} that the stabilizers of $\text{SHP}(H_1,H_2)$ are exactly those stabilizers of $\text{HGP}(H_1,H_2)$ that are supported entirely on the large lattice. Similarly, one can check that the stabilizers of $\text{SHP}(H_2^T,H_1^T)$ are those of $\text{HGP}(H_1,H_2)$ supported entirely on the small lattice (i.e.~Eqs.~(\ref{eq:small_lattice_Sx},\ref{eq:small_lattice_Sz})). Also, Eq.~\eqref{eq:count_gauge_qubits} says that $\text{SHP}(H_2^T,H_1^T)$ has $(n_2^T-k_2^T)(n_1^T-k_1^T)$ gauge qubits, which is the same number as $\text{SHP}(H_1,H_2)$ by Eq.~\eqref{eq:parameter_relation}.

These two facts suggest the following theorem.
\begin{thm}
$\mathcal{Q}'=\text{HGP}(H_1,H_2)$ is a gauge-fixing of $\mathcal{Q}=\text{SHP}(H_1,H_2)\text{SHP}(H_2^T,H_1^T)$.
\end{thm}
\begin{proof}
We employ Definition~\ref{defn:gauge_fixing}. It should be clear that
\begin{align}
K(\mathcal{Q})&=K(\text{SHP}(H_1,H_2))+K(\text{SHP}(H_2^T,H_1^T))\\
&=k_1k_2+k_1^Tk_2^T=K(\mathcal{Q}'),
\end{align}
therefore satisfying part (2) of the definition.

For part (1), it is important to associate (via a 1-1 map) the physical qubits of $\mathcal{Q}$ and $\mathcal{Q}'$. Recall, qubits of $\mathcal{Q}'$ are placed on the two lattices $L$ and $l$. A qubit at site $(i,j)$ in $\text{SHP}(H_1,H_2)$ is associated with the qubit at $(i,j)$ in $L$. On the other hand, qubit $(i,j)$ of $\text{SHP}(H_2^T,H_1^T)$ is associated instead with the qubit at $(j,i)$ on the small lattice $l$. Now, taken as a whole, this code $\mathcal{Q}$ has gauge operators and stabilizers that can be written as
\begin{widetext}
\begin{align}
\mathcal{G}^{(\mathcal{Q})}_X&=\left\{X^{(L)}(S)X^{(l)}(T):G_1S=0,TF_2^T=0\right\},\\
\mathcal{G}^{(\mathcal{Q})}_Z&=\left\{X^{(L)}(S)X^{(l)}(T):SG_2^T=0,F_1T=0\right\},\\
\mathcal{S}^{(\mathcal{Q})}_X&=\{X^{(L)}(S)X^{(l)}(T):SH_2^T=0,H_1^TT=0,G_1S=0,TF_2^T=0\},\\
\mathcal{S}^{(\mathcal{Q})}_Z&=\{Z^{(L)}(S)Z^{(l)}(T):H_1S=0,TH_2=0,SG_2^T=0,F_1T=0\}.
\end{align}
\end{widetext}
It should now be clear that we have
\begin{align}
\mathcal{S}_X^{\mathcal{Q}}&\le\mathcal{S}_X^{\text{hgp}}\le\mathcal{G}_X^{\mathcal{Q}},\\
\mathcal{S}_Z^{\mathcal{Q}}&\le\mathcal{S}_Z^{\text{hgp}}\le\mathcal{G}_Z^{\mathcal{Q}},
\end{align}
thereby satisfying part (1) of Def.~\ref{defn:gauge_fixing}.
\end{proof}
In essence, the two SHP codes live on the large and small lattices in Fig.~\ref{fig:hypergraph_prod.png}, respectively, and gauge-fix to the HGP code by placing their gauge qubits in $(n_1-k_1)(n_2-k_2)$ maximally entangled two-qubit states.

\section{Decoding BBS and SHP codes}
\label{sec:decode}
Both the BBS codes and SHP codes can be decoded by directly running a classical decoder on the corresponding classical code used to construct the quantum code. In this section we review the decoding of BBS codes, as discussed in \cite{yoder2019optimal}, and show that similar arguments can be applied to SHP codes. We review the classical belief propagation decoder for expander codes, and show how it can be used to tolerate measurement errors and therefore decode BBS and SHP codes. 

\subsection{Decoding the BBS codes}
To decode the BBS codes, we have to establish associations between the stabilizers of the quantum code and the parity checks of the classical code. For convenience, we assume that $A$ is an $n \times n$ symmetric matrix constructed as $A = G^T Q G$, where $G$ is the generating matrix of a $[n,k,d]$ classical code $\mathcal{C}$, so there is only one classical code $\mathcal{C}=\row{A}=\col{A}$ under consideration. We let $H$ be the parity check matrix of $\mathcal{C}$. 

Given Eq. \ref{eq:BBS_Sx}, let $S$ be the support of a X-type stabilizer of BBS(A), $X(S \cap A) \in \mathcal{S}_X^{(\text{bbs})}$. Since $SH^T_R=0$, rows of $S$ are codewords of $\mathcal{C}_R$, either all 1s or all 0s. Because $GS=0$, columns of $S$ are parity checks of $\mathcal{C}$. Therefore, $S \cap A = \diag{\vec{r}}A$ for some $\vec{r} \in \row{H}$. Hence we have
\begin{equation}
    \mathcal{S}_X^{(\text{bbs})} = \{ X(\diag{\vec{r}}A):\vec{r}\in \row{H} \}.
\end{equation}
Similarly, 
\begin{equation}
    \mathcal{S}_Z^{(\text{bbs})} = \{ Z(A\diag{\vec{c}}):\vec{c}\in \row{H} \}.
\end{equation}
Thus, the parity checks of the classical code indicate which sets of rows or columns constitute a stabilizer, and give us a one-to-one correspondence between the quantum stabilizers and the classical parity checks.

Since single qubit Pauli $X$ errors within a column are equivalent up to gauge operators, each column is only sensitive to an odd number of Pauli $X$ error. The even or oddness of a column corresponds to the 0 or 1 state of an effective classical bit in the code $\mathcal{C}$. Similarly, the symmetry of A indicates that the same correspondence holds for Pauli $Z$ errors in rows and the even or oddness of rows in $A$.

\begin{algo}[\textbf{The Induced Decoder for $\text{BBS}(A)$}]
Given a symmetric binary matrix $A=G^T Q G$ where $\mathcal{C}=\row{G}=\row{A}=\col{A}$ is a classical $[n,k,d]$ code, we can decode the Bravyi-Bacon-Shor code $\text{BBS}(A)$ by:\\

\begin{itemize}
    \item Collect the $X$- or $Z$-type syndrome $\vec{\sigma}$ for the quantum code $\text{BBS}(A)$.
    \item Run the classical decoder to obtain a set of corrections for the classical code $\vec{c}=\mathcal{D}(\vec{\sigma})$.
    \item For each bit in the correction $\vec{c}$, apply a Pauli $Z$- or $X$-type correction to a single qubit in each row or column corresponding to the classical bit.
\end{itemize}
\end{algo}
The time complexity of the induced decoder consists of the time to construct the stabilizer values and the time to run the classical decoder $\mathcal{D}$. Given an $[n,k,d]$ classical code $\mathcal{C}$, for a weight-$w$ parity check of the classical code the corresponding stabilizer of the $\text{BBS}(A)$ code is the sum of $O(wn)$ two-qubit gauge measurement. There are $m$ such stabilizers and suppose the classical decoder runs in time at most $t$, then the induced decoder of $\text{BBS}(A)$ takes time $O(mwn+t)$. When classical expander codes are used to construct $\text{BBS}(A)$ and the belief propagation decoder is used as classical decoder $\mathcal{D}$, $m=O(n)$, $w=O(1)$, $t=O(n)$, so the induced decoder runs in time $O(mwn+t)=O(n^2+n)=O(N)$, which is linear in the size of the quantum code.

\subsection{Decoding the SHP codes}
Similar to what we have done for the BBS codes, to decode the SHP codes we have to associate the stabilizers of the quantum code to the parity checks of the classical codes. To illustrate the idea most easily, we assume that the X and Z part of the SHP code are generated by the same $[n,k,d]$ classical code $\mathcal{C}$ with generating matrix $G \in \mathcal{F}_2^{k \times n}$ and parity check matrix $H \in \mathcal{F}_2^{m \times n}$, SHP($H_1, H_2$) = SHP($H$). From Eq.~\eqref{eq:shp_SX} and Eq.~\eqref{eq:shp_SZ} we have 
\begin{align}
    S_Z^{(\text{shp})} &= Z(G \otimes H)\\
                       &= \{ Z(g^T \otimes h):g\in \row{G}, h\in \row{H} \},
\end{align}
\begin{align}
    S_X^{(\text{shp})} &= X(H \otimes G)\\
                       &= \{ X(h^T \otimes g):h\in \row{H}, g\in \row{G}\}.
\end{align}
Since $\rank{G}=k$, the eigenvalues of the quantum stabilizers correspond to exactly $k$ sets of syndromes for the classical code $\mathcal{C}$. In the case of $Z$-type stabilizers, let $g_i$ be the $i$-th row of $G$ for each $i$ such that $1 \leq i \leq k$. A set of syndromes for the classical code $\mathcal{C}$ is generated by measuring the following set of stabilizers
\begin{equation}
    \{ Z(g_i^T \otimes h_j): h_j \in \row{H}, 1 \leq j \leq m \}.
\end{equation}
These $k$ sets of syndromes are passed to the classical decoder and results in $k$ sets of $n$-bit corrections on $\mathcal{C}$. However, in order to apply these $k$ sets of classical corrections canonically onto independent sets of qubits in the SHP code without affecting each other, we have to make sure that the generating matrix $G$ is in the reduced row echelon form, so that the $i$-th set of corrections can be applied on the $i$-th row of qubits lattice.

\begin{algo}[\textbf{The Induced Decoder of $\text{SHP}(H)$}]
Given a $[n,k,d]$ classical code $\mathcal{C}$ with generating matrix $G \in \mathcal{F}_2^{k,n}$ and parity check matrix $H \in \mathcal{F}_2^{m,n}$, the hypergraph subsystem code SHP(H) can be decoded by
\begin{itemize}
    \item Reshape $G$ into its reduced row echelon form \\
    $G = [I_k \,\,B]$.
    \item Collect the $X$- or $Z$-type syndrome $\vec{\sigma}$ of the quantum code $\text{SHP}(H)$.
    \item For each $i \in \{1,2,\ldots,k\}$, the syndrome corresponding to the set of stabilizers $\{ Z(g_i^T \otimes h_j): h_j \in \row{H}, 1 \leq j \leq m \}$ is passed to the classical decoder $\mathcal{D}$, and a $n$-bit correction $\vec{c_i}$ is obtained.
    \item For each set of corrections $\vec{c_i}$, Pauli $Z$- or $X$-type corrections are applied to the qubits $\vec{e_i} \otimes [1,1, \ldots, 1]$, where $\vec{e_i}$ is the $i$-th unit vector.
\end{itemize}
\end{algo}
Hence when using the induced decoder on the quantum code $\text{SHP}(H)$ that is constructed by a $[n,k,d]$ code, we have to run the classical decoder $\mathcal{D}$ a total of $k$ times. The correction consists of $k$ sets of $n$-qubit Paulis have to be applied on the first $k$ rows or columns of the qubit lattice.

The time complexity of the induced decoder for SHP codes again consists of the time to construct the stabilizer values and the time to run the classical decoder $\mathcal{D}$. For a weight-$w$ parity check of the classical code, the corresponding stabilizer of the SHP code is the sum of $O(n)$ number $w$-qubit gauge measurements. There are $O(k \times m)$ stabilizers, and the classical decoder $\mathcal{D}$ needs to be run $k$ times where each run takes time at most $t$, then the induced decoder takes time $O(kmnw+kt)$. When classical expander codes are used to construct the SHP code and the belief propagation decoder is used as the classical decoder $\mathcal{D}$, $m = O(n)$, $k = O(n)$, $w = O(1)$, $t = O(n)$, so the induced decoder runs in time $O(kmnw + kt)=O(n^3+nt)=O(N^{3/2})$.

\subsection{Classical belief propagation decoder}
In the previous two sections we have shown that decoding both the BBS codes and the SHP codes amount to directly decoding the underlying classical code $\mathcal{C}$ that was used to construct the quantum code, and apply the resulting corrections to the appropriate set of qubits in the quantum code. Therefore, in order to maximize the performance of the induced decoding algorithm the best classical decoder should be employed with modifications to tolerate measurement noise. Sipser and Spielman have analyzed the flip decoder \cite{sipser1996expander, spielman1996linear} for classical expander codes and in the scenario that the parity checks are noisy in addition to the bits. A quantum version of the classical flip decoder has been shown to decode the quantum expander codes efficiently \cite{leverrier2015quantum, fawzi2018efficient, fawzi2018constant, leverrier2015quantum}.

However, when classical LDPC codes and expander codes are considered, various iterative message-passing decoding algorithms have been shown to result in codes with rate approaching the Shannon capacity together with efficient decoding algorithm (see e.g.~\cite{richardson2000design}). Message passing algorithms get the name as information is transmitted back and forth between variable and check nodes along the edges of the graph that is used to define the classical code. The transmitted message along an edge is a function of all received messages at the node except for a particular edge. This property ensures that the incoming messages are independent for a tree like graph. Among these well-known decoding algorithms, the belief propagation (BP) decoder, sometimes referred to as Gallager's soft decoding algorithm \cite{Gallager1962}, have been shown to out perform other message-passing algorithms for classical LDPC codes when the binary symmetric channel (BSC) is considered. In this section we briefly describe the BP decoder for classical LDPC codes. For a comprehensive discussion of this area, we point the reader to the book by Richardson and Urbanke \cite{richardson2008modern} and the notes by Guruswami \cite{guruswami2006iterative}, which are excellent resources on this topic.

In particular, here we present the modified BP decoder that uses parity check values as input instead of bit values, in order to simulate the quantum case where data qubit values are not known to decoders. In order to run the BP decoder using parity check values, we add another set of $m$ ``syndrome nodes" $s_j, 1 \leq j \leq m$, that have one-to-one correspondence to the check nodes: syndrome node $s_j$ and check node $j$ are connected by edge $(s_j, j)$. These syndrome nodes $s_i$ are used to store the measured parity check values. Without loss of generality, we assume that the all $0$s message is the correct message to be received.
\begin{algo} [\textbf{Belief Propagation Decoding Algorithm}]
Assuming the probability $p$ for each bit of the incoming message to be flipped is the same, then the log-likelihood ratio $m_i$ of the $i$-th bit is
\begin{equation}
    m_i = \log{\frac{1-p}{p}}.
\end{equation}
For the syndrome nodes $s_i$, we let $m_{s_i} = + \infty$ if the $i$-th syndrome is $0$ and $m_{s_i} = - \infty$ if the $i$-th syndrome is $1$.
Do the following two steps alternatively: 
\begin{enumerate}
    \item \textbf{Rightbound messages}: For all edges $e = (i,j)$, $i \in \{1, 2, \ldots, n\} \cup \{s_1, s_2, \ldots, s_m \}$, do the following: 
    \begin{itemize}
        \item if this is the zeroth round, $g_{i,j} = m_i$.
        \item Otherwise
            \begin{equation}
                g_{i,j} = m_i + \sum_{k \in \mathcal{N}(i) \backslash j}h_{i,k}
            \end{equation}
            where $\mathcal{N}(i)$ denotes the set of neighbors of node $i$. 
    \end{itemize}
     The variable node $i$ sends the message $g_{i,j}$ to check node $j$.
     
     \item \textbf{Leftbound messages}: For edges $e = (i,j)$, $i \in \{1, 2, \ldots, n \}$ do the following:
     \begin{equation}
         h_{i,j} = f \left( \prod_{k \in \mathcal{N}(j) 
        \backslash i} \frac{e^{g_{k,j}}-1}{e^{g_{k,j}}+1} \right), \,\,\,\, f(u) = \log{\frac{1+u}{1-u}}.
     \end{equation}
     The check node $j$ sends the message $h_{i,j}$ to node $i$.
\end{enumerate}
At each step we can determine the current variable node values $v_i$ given their updated log-likelihood ratios: $v_i = 0$ if $m_i > 0$ and $v_i = 1$ if $m_i < 0$. The above iterative step terminates when all check nodes are satisfied based on the current $v_i$, or the predetermined number of iterations is reached. The variable node value $v_i$ at the final step is used as correction for the noisy channel output $b_i$.

\label{algo:BP}
\end{algo}
If the graph considered has large enough girth when compared to the number of iterations of the algorithm, the messages at each iteration would approach the true log-likelihood ratio of the bits given the observed values. By applying expander graph arguments to message passing algorithms it has been shown that the BP decoding algorithm can correct errors efficiently, with time linear in the block size \cite{burshtein2001expander}. Therefore the belief propagation decoder is a good candidate for decoding the BBS and SHP codes constructed using classical LDPC codes.

\subsection{Handling measurement errors with BP decoder}
As we mentioned previously, decoding algorithms for classical codes usually do not consider the problem of measurement noise. In previous studies of the BP decoder, no explicit proposals have been made regarding handling measurement noise when decoding classical expander codes. In order to use the BP decoder to decode the BBS and SHP codes as part of the induced decoder, modifications have to be made in order to tolerate measurement errors on parity check measurements.

When given a classical code $\mathcal{C}$ with $n$ variable nodes and $m$ check nodes, in addition to what we have done in Algorithm \ref{algo:BP} we add $m$ variable nodes to the graph so that each of them has a one-to-one correspondence with the $m$ check nodes: variable node $n+j$ is connected to check node $j$ via edge $(n+j, j)$. These additional variable nodes are used to represent measurement errors on the parity checks. An example of the modified graph for decoding a classical linear code of block length $6$ is shown in FIG. \ref{fig:BP_decoding}. In the binary symmetric channel, let $p$ be the probability that a bit is flipped and let $q$ be the probability that a measurement is flipped. We define the log-likelihood ratio $m_i$ for the $n+m$ variable nodes as:
\begin{itemize}
    \item 
        $m_i = \log{\frac{1-p}{p}}, 1 \leq i \leq n$,
    
    \item 
        $m_i = \log{\frac{1-q}{q}}, n+1 \leq i \leq n+m$.
\end{itemize}
For the syndrome nodes $s_i$, we let $m_{s_i} = + \infty$ if the $i$-th syndrome is $0$ and $m_{s_i} = - \infty$ if the $i$-th syndrome is $1$.

When executing the belief propagation decoding algorithm, the normal message passing process is executed as described in Algorithm \ref{algo:BP}. For the added syndrome node $s_i$, they send their log-likelihood ratios $m_i$ to the associated check node $j$ with message $g_{i,j} = m_i$ during the rightbound messages phase in each iteration, but there will be no incoming messages from the check nodes $c_i$ to change their values. The algorithm terminates when all check nodes are satisfied or a predetermined number of iterations is reached, and the $n$-bit variable node values at the final step are used as corrections for the noisy data qubits.
\begin{figure}
    \centering
    \includegraphics[width = 0.9\linewidth]{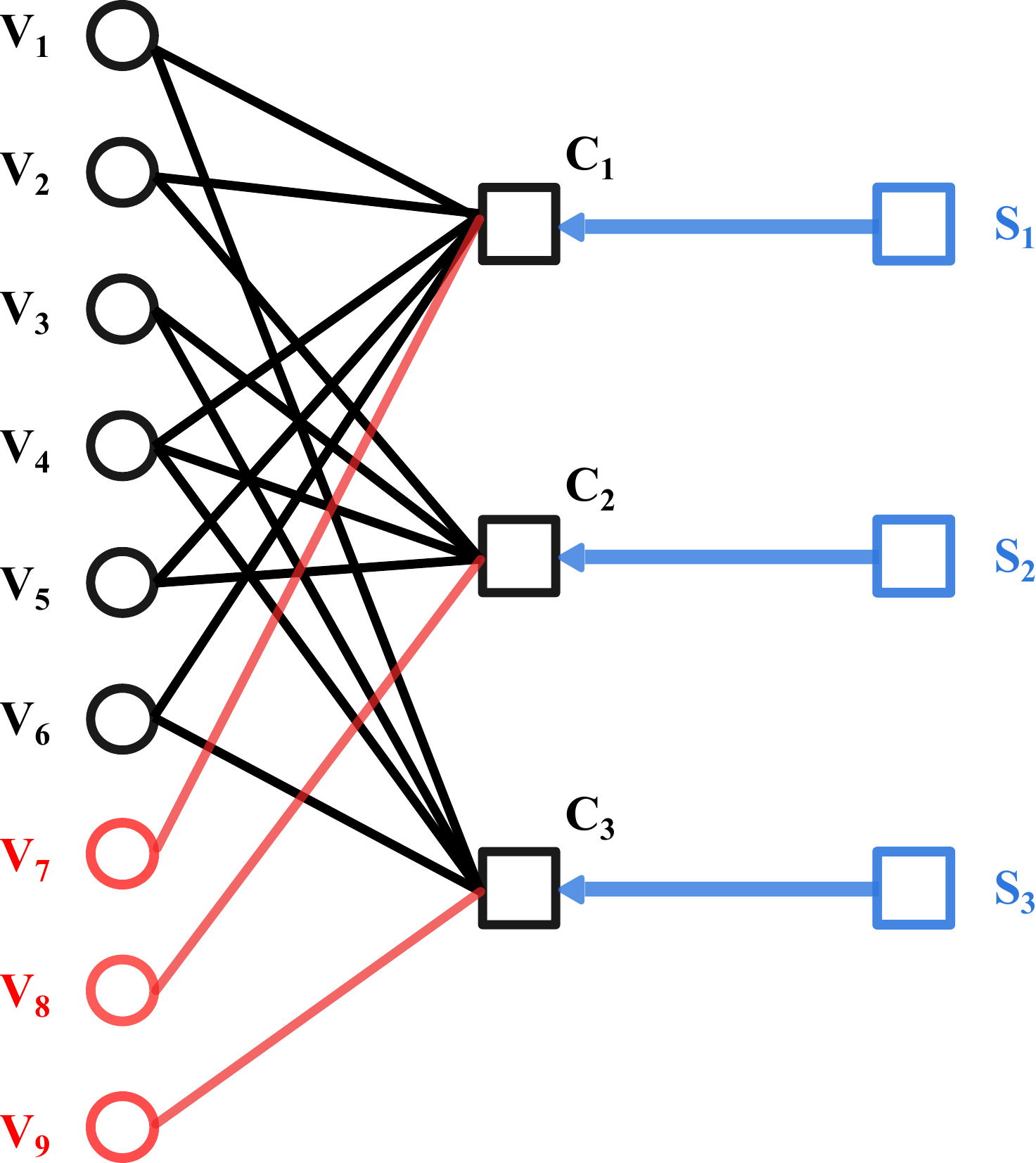}
    \caption{The graph for decoding a classical code of length $6$ using the modified BP decoder that tolerates measurement errors. The syndrome nodes $s_1, s_2, s_3$ are assigned log-likelihood values $\pm \infty$ given the input parity check measurement values $0$ or $1$.}
    \label{fig:BP_decoding}
\end{figure}

By employing the above described modifications to the BP algorithm, we can efficiently decode the classical expander codes while tolerating measurement errors.

\section{Numerical Simulations and Results}
\label{sec:results}
\begin{figure}
\centering
    \begin{subfigure}[b]{0.45\textwidth}
        \centering
        \includegraphics[width = \linewidth]{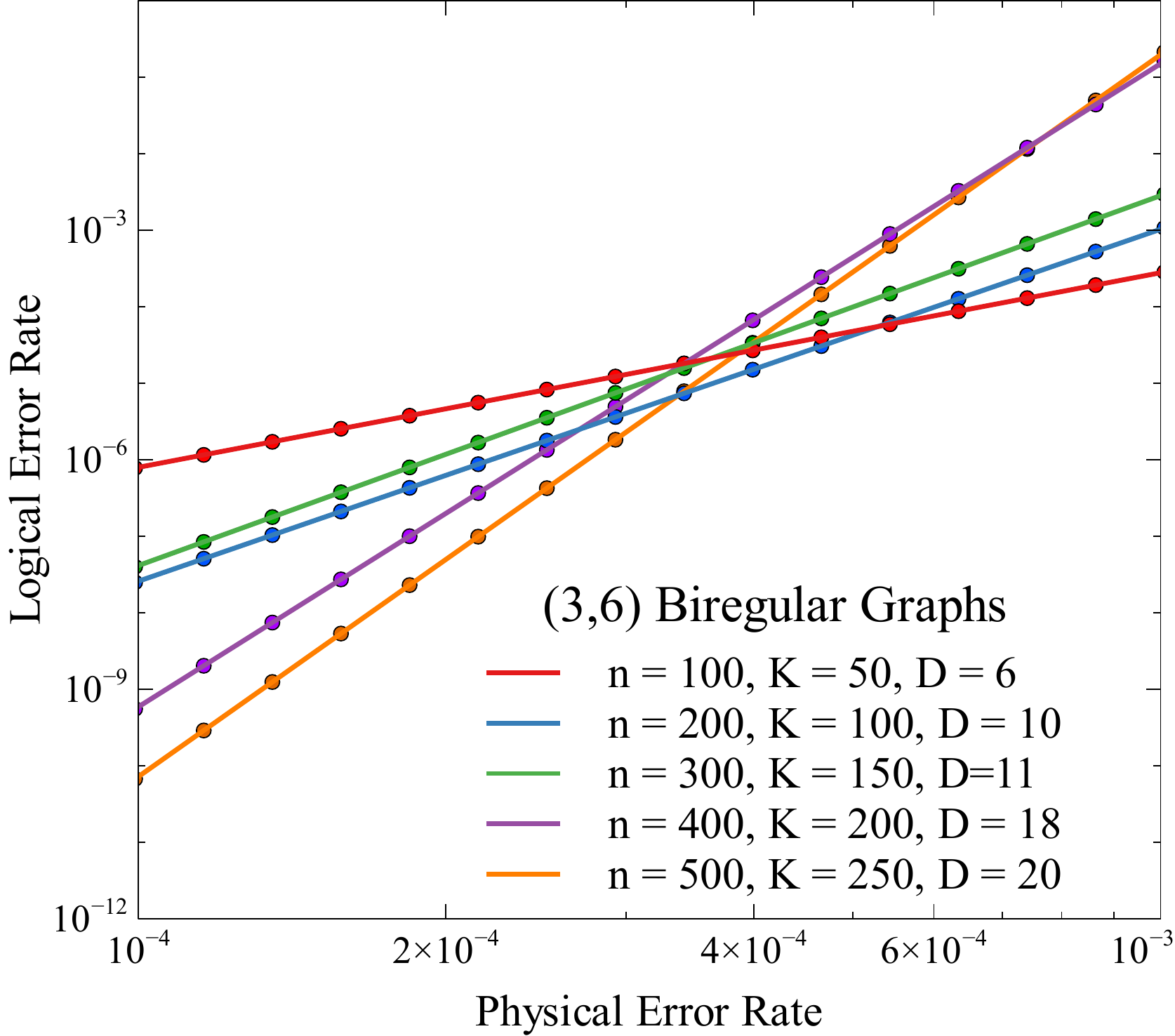}
        \caption{}
    \label{fig:BBS_36}
    \end{subfigure}%
    \hfill%
    \begin{subfigure}[b]{0.45\textwidth}
        \centering
        \includegraphics[width= \linewidth]{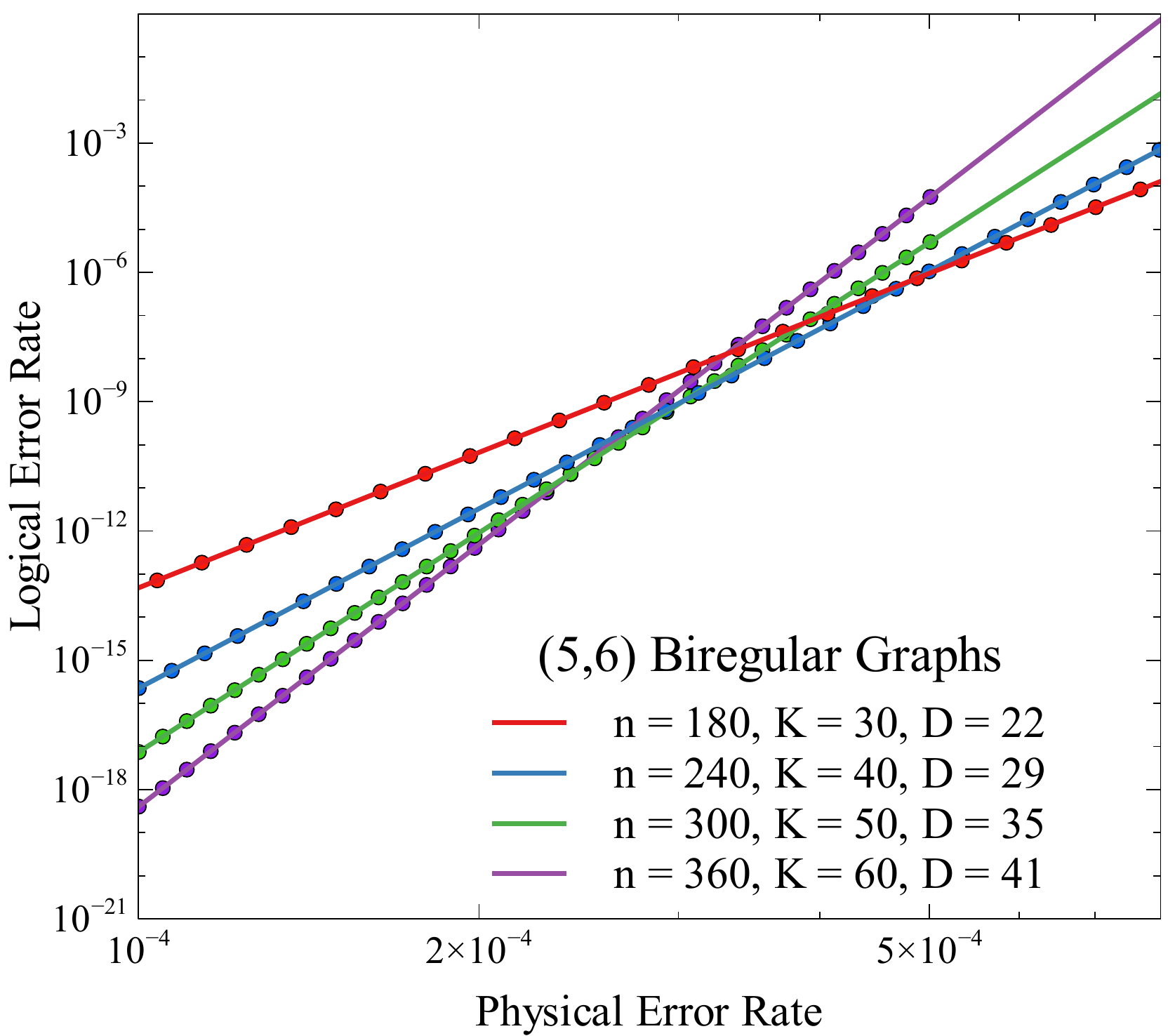}
        \caption{}
    \label{fig:BBS_56}
    \end{subfigure}
\captionsetup{justification=raggedright}
\caption{Simulating the performance of Bravyi-Bacon-Shor codes constructed using (a) $(3,6)$- and (b) $(5,6)$-biregular bipartite graphs. BBS codes by $(5,6)$ graphs outperforms BBS codes by $(3,6)$ graphs due to superior performance of classical $(5,6)$ codes.}
\label{fig:BBS}
\end{figure}

In this section we present numerical results of decoding the BBS and SHP codes using the induced decoders instantiated with the modified BP decoder that handles measurement errors. All simulations are done under the phenomenological error model, where given probability $p$, random single-qubit bit or phase flip errors of the form $E_{1q} = \{\sqrt{1-p}I, \sqrt{p}X\}$ or $E_{1q} = \{\sqrt{1-p}I, \sqrt{p}Z\}$ are applied independently on qubits and measurements output the wrong (opposite) value with probability $p$. There is no circuit-level error propagation in the simulation. 

In order to maximize the parameters and performance of the quantum codes when decoded by the induced decoders, we construct the BBS and SHP codes with classical regular LDPC codes defined by biregular bipartite graphs. To obtain symmetric performance for $X$- and $Z$-type errors, both $X$ and $Z$ part of each quantum code are constructed with the same classical LDPC code. Since both the BBS and SHP codes are defined as CSS codes, $X$- and $Z$-type errors can be decoded separately using the induced decoder. Hence in the rest of the paper we assume that each qubit independently suffers from Pauli $X$- and $Z$-type errors as described in the previous paragraph, and study the performance of these codes by plotting the average logical error rate per logical qubit of the $K$-qubit block versus the physical error rate of each qubit. Using this metric allows us to directly compare the average performance of quantum codes with different encoding rates on an equal footing, instead of comparing large blocks with vastly different numbers of encoded qubits. By doing so we are taking into account both the performance and encoding rate when comparing different codes, but to some extent ignoring the potential correlation between logical errors.

The classical regular LDPC codes that are used to construct the BBS and SHP codes were randomly generated biregular bipartite graphs using the configuration model \cite{richardson2008modern}. It can be shown that asymptotically these graphs will have a good expansion coefficient, making them classical expander codes with good performance. For each of the selected block size, we randomly generated $1000$ biregular bipartite graphs with specified node degrees and simulated their performance under the binary symmetric channel. The best-performing classical code is chosen to construct the quantum code. Since the induced decoder for the quantum code directly decodes on the underlying classical code, a relatively good classical code implies a relatively good quantum code.

\begin{figure}
    \centering
    \includegraphics[width = 0.9\linewidth]{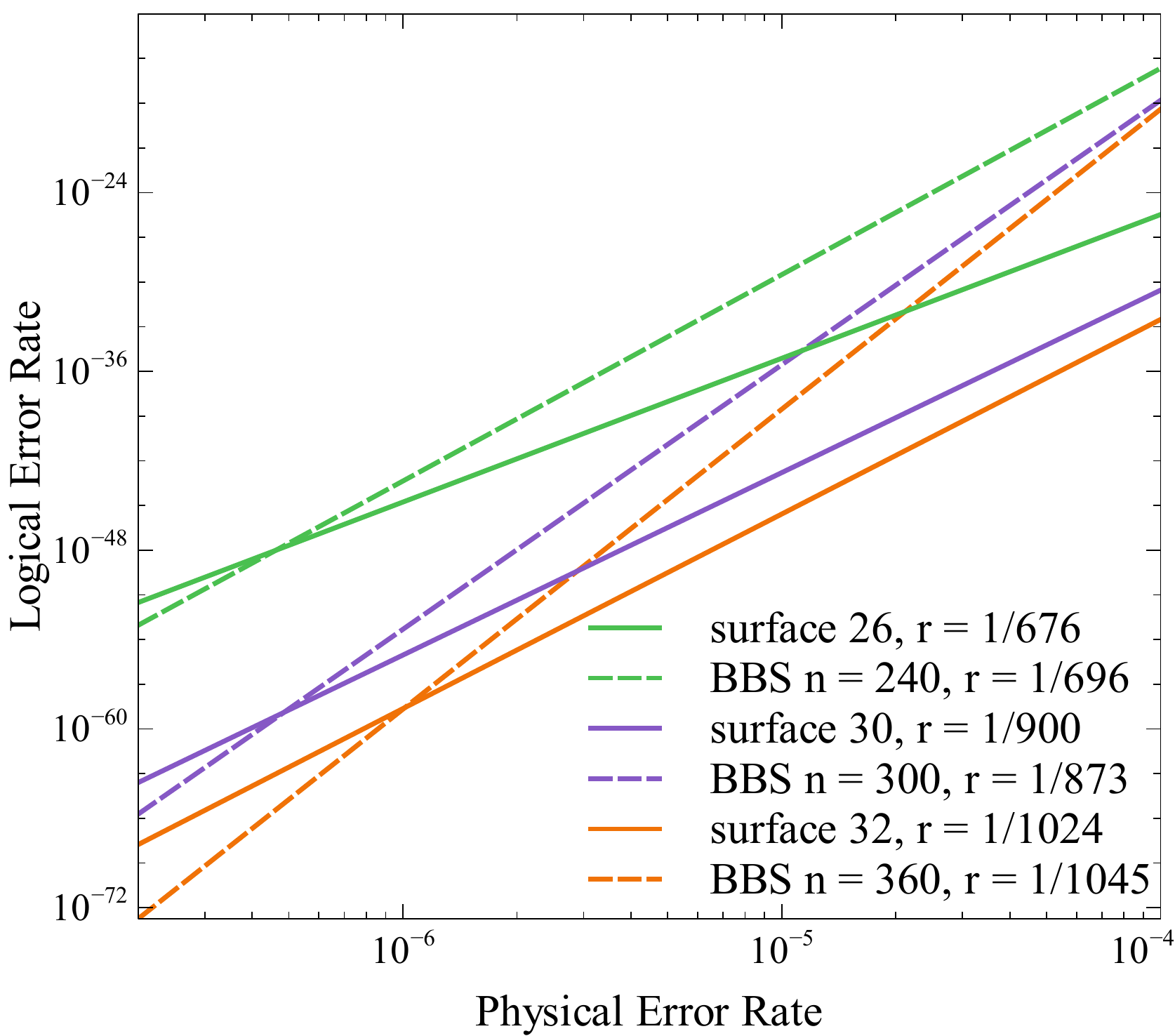}
    \caption{Comparing the average error rate per logical qubit of the BBS codes constructed with (5,6)-biregular bipartite graphs of block size $240, 300, 360$ to surface codes of sizes $26 \times 26, 30 \times 30, 32 \times 32$. }
    \label{fig:BBS_VS_SF}
\end{figure}

We studied two classes of graphs for generating classical LDPC codes: the $(3,6)$- and $(5,6)$- biregular bipartite graphs, which we will refer to as the $(3,6)$ and $(5,6)$ codes. By simulating the performance of these two classical codes with the BP decoder, we observed that the $(5,6)$ codes significantly outperform the $(3,6)$ codes, which agrees with previous studies in classical coding theory \cite{richardson2001capacity, mackay1999good}. Given a $(b,c)$ code of size $n$, the number of encoded bits is $k=\frac{c-b}{b}n$ and the encoding rate for the classical code is $\frac{c-b}{b}$. Hence the BBS codes constructed with $(b,c)$ classical code have parameters $\llbracket N_{BBS}, K_{BBS} \rrbracket = \llbracket O(n^2), \frac{c-b}{b} n \rrbracket$, and the SHP codes have parameters $\llbracket N_{SHP}, K_{SHP} \rrbracket = \llbracket n^2, (\frac{c-b}{b})^2 n^2 \rrbracket$.

In all plots, $n$ is the number of bits/variable nodes for the classical LDPC code, $N$ is the number of physical qubits in the quantum code, $K$ is the number of encoded logical qubits in the quantum code, $D$ is the average distance of the quantum code found through fitting the simulated data to $P_L=A p^D$, and $r$ is the encoding rate of the quantum code. The numerical performance of the BBS codes presented in Figures \ref{fig:BBS}, \ref{fig:BBS_VS_SF} and \ref{fig:BBS_VS_SHP} are obtained using importance sampling to error rates as low as $10^{-4}$, and best-fit lines are plotted in order to extrapolate the codes behavior to low error regimes. Details of importance sampling can be found in \cite{LiBareAnc2017}. The numerical performance of the SHP codes presented in Figures \ref{fig:SHP} and \ref{fig:BBS_VS_SHP} are obtained through Monte Carlo simulations at various physical error rates.

From FIG. \ref{fig:BBS} we can see that the BBS codes constructed with $(5,6)$ codes have significantly better performance than that with $(3,6)$ codes, as expected given the results on the classical codes. It is clear from FIG. \ref{fig:BBS} that the BBS codes do not have a fault-tolerant threshold, due to the fact that the weight-$2$ gauge operators in the quantum code result in a superexponential scaling of number of weight-$D$ dressed logical operators. A similar behavior is observed for the SHP codes constructed with $(5,6)$ codes, as shown in \ref{fig:SHP}, where the SHP codes also do not exhibit a fault-tolerant threshold. 

\begin{figure}
    \centering
    \includegraphics[width = 0.9\linewidth]{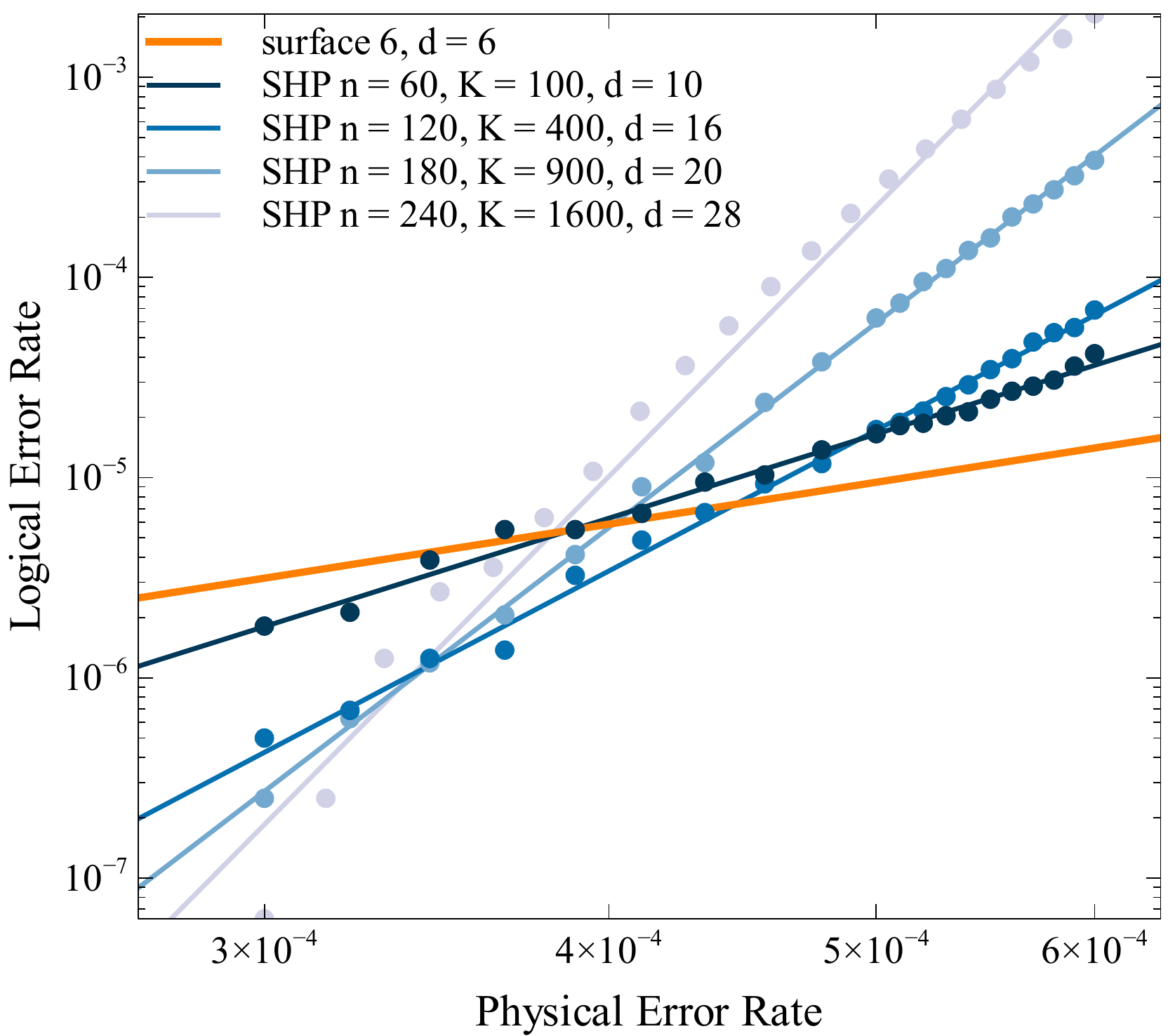}
    \caption{Simulating the performance of the SHP codes constructed using $(5,6)$-biregular bipartite graphs. Their average performance per logical qubit are compared to the size $6 \times 6$ surface code. All codes in this plot have encoding rate $1/36$.}
    \label{fig:SHP}
\end{figure}

\begin{figure}
    \centering
    \includegraphics[width = 0.9\linewidth]{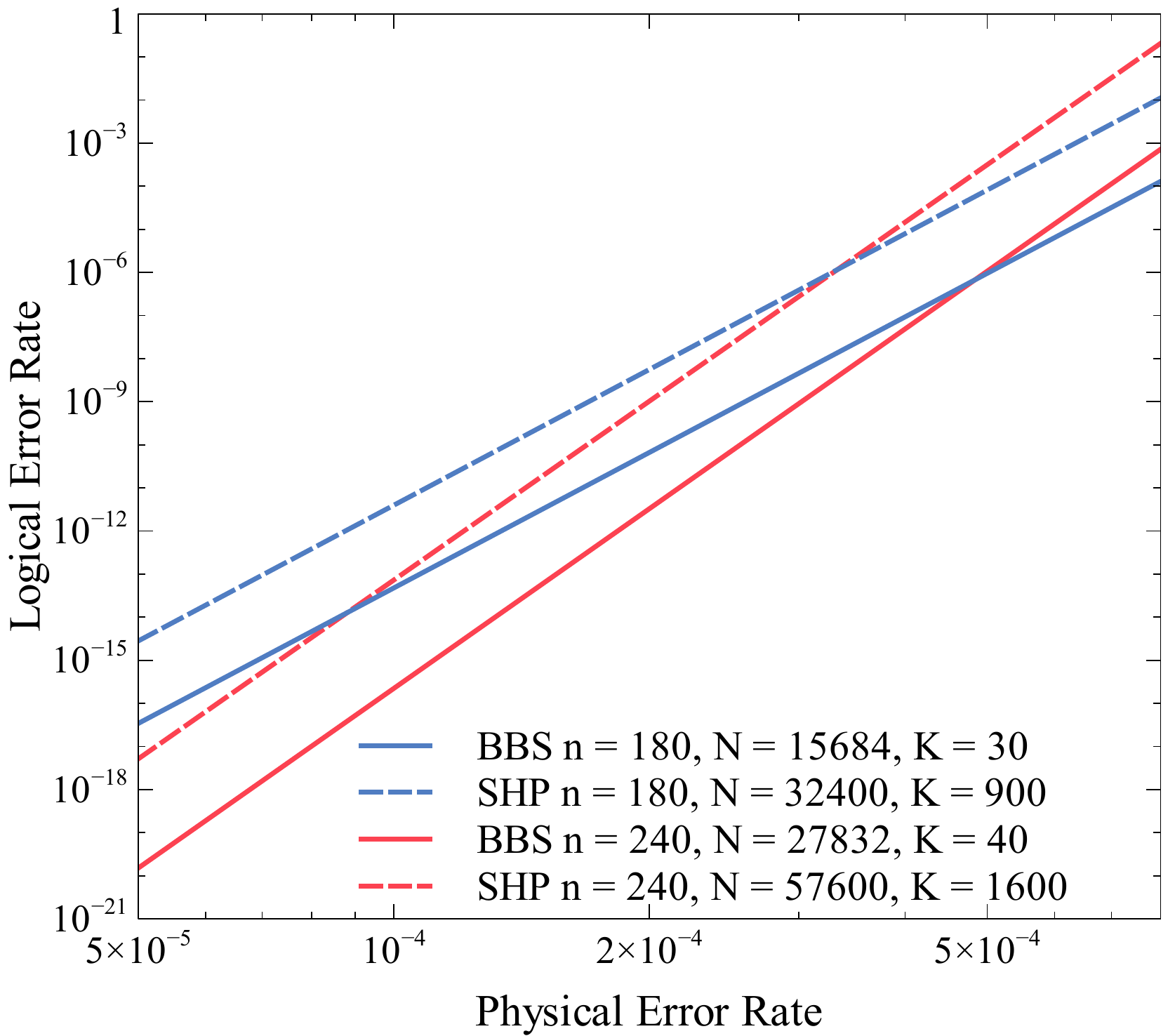}
    \caption{Comparing the average performance per logical qubit of the BBS codes to the SHP codes. The BBS and SHP codes with the same $n$ are constructed using the same $(5,6)$-biregular bipartite graph.}
    \label{fig:BBS_VS_SHP}
\end{figure}

To benchmark the performance of the BBS codes and SHP codes, we compare them to the surface codes as well as to each other. As previously mentioned, in order to obtain a reasonable comparison we compare the BBS and SHP codes against surface codes of similar encoding rate $r$ by comparing the average error rate of single logical qubits within the same code block to the logical error rate of the surface code. In FIG. \ref{fig:BBS_VS_SF} we are comparing the average logical error rate per logical qubit of the BBS codes constructed with $(5,6)$ codes of sizes $n=240, 300, 360$ to surface codes of sizes $26\times 26, 30 \times 30, 32 \times 32$. Note that the surface code results are simulated using the Union Find decoder \cite{delfosse2017almost}, so that we are comparing a linear time decoding algorithm of the BBS codes to a linear time decoding algorithm of the surface code. The BBS codes have better distances than surface codes of similar encoding rates, but they only outperform the surface codes for physical error rates below $10^{-6}$. Similar results for SHP codes are shown in FIG.\ref{fig:SHP}. Since the SHP codes have constant encoding rates and when the $(5,6)$ codes are used, the resulting encoding rate is $r=1/36$, so we are comparing the average error rate for single logical qubits in the SHP codes to a single $6 \times 6$ block of surface code. The SHP codes can have significantly better distance than surface code of the same encoding rate, but they do not outperform the surface code until physical error rates $p \leq 4 \times 10^{-4}$.

Finally, we compare the average performance per logical qubit of the BBS and SHP codes, as shown in \ref{fig:BBS_VS_SHP}. The comparison is made between BBS and SHP codes constructed using the exact same $(5,6)$-biregular bipartite graph. While it seems that the SHP codes' average logical qubit performance is slightly worse than that of the BBS codes, bear in mind that the SHP codes have much higher encoding rate.

\section{Conclusion}
We studied two different constructions of quantum subsystem error-correcting codes using classical linear codes: the Bravyi-Bacon-Shor (BBS) codes and the subsystem hypergraph product (SHP) codes. We reviewed the BBS codes that was introduced in a previous paper \cite{yoder2019optimal}, and presented a construction of the SHP codes that can be viewed similar to the hypergraph product codes \cite{tillich2013quantum}. We proposed efficient algorithms to decode the BBS and SHP codes while handling measurement errors by using a modified belief propagation decoder for classical expander codes. We studied the numerical performance of the BBS and SHP codes, and showed that while these codes do not have a fault-tolerant threshold, they have very good distance scaling and encoding rates. 

When constructed using classical expander codes, the BBS codes have encoding rates $O(1/\sqrt{N})$ and the SHP codes have constant encoding rates dependent on the expander code parameters. Suppose the same classical expander code is used, the resulting SHP codes have even higher encoding rates than the hypergraph product codes. Hence for large block sizes these codes could offer significant savings in terms of resource overhead when trying to achieve a specific logical error rate. It is worth noting that while we are already observing very good logical performance by simulating codes constructed with small biregular bipartite graphs, classical LDPC codes asymptotically become better expander codes and the belief propagation decoder will give a much better performance for expander codes of larger block sizes.

Therefore, the BBS and SHP codes are worth studying for the purpose of large scale quantum error correction. Future studies on these codes could include investigating the potential of using large irregular LDPC codes to construct the BBS and SHP codes in order to achieve better logical performances, tailoring the quantum code for biased noise models by using two different classical codes to construct asymmetric BBS and SHP codes, and methods to apply fault-tolerant logical operations within the same code block.

\section{Acknowledgement}
The authors gratefully acknowledge Andrew Cross, Leonid Pryadko, Ken Brown, Michael Newman and Dripto Debroy for helpful discussions. T.~Yoder also thanks the IBM Research Frontiers Institute for partial support. M.~Li thanks the IBM graduate internship program and the National Science Foundation Expeditions in Computing award 1730104 for support.

\bibliography{references.bib}

\end{document}